\documentclass[preprint,12pt]{elsarticle}

\usepackage{amsmath,amssymb,amsthm}
\usepackage{graphicx}
\usepackage{tikz}
\usetikzlibrary{arrows.meta,positioning}
\usepackage[hidelinks]{hyperref}

\journal{Physica D: Nonlinear Phenomena}

\newtheorem{theorem}{Theorem}[section]
\newtheorem{proposition}[theorem]{Proposition}
\newtheorem{lemma}[theorem]{Lemma}
\newtheorem{corollary}[theorem]{Corollary}
\theoremstyle{definition}
\newtheorem{definition}[theorem]{Definition}

\theoremstyle{remark}
\newtheorem{remark}[theorem]{Remark}

\numberwithin{equation}{section}
\newcommand{\C}{\mathbb{C}}
\newcommand{\Z}{\mathbb{Z}}
\newcommand{\PP}{\mathbb{P}}
\newcommand{\LL}{\mathcal{L}}
\newcommand{\EE}{\mathcal{E}}
\newcommand{\cA}{\mathcal{A}}
\newcommand{\cV}{\mathcal{V}}
\newcommand{\cD}{\mathcal{D}}
\newcommand{\bu}{\mathbf{u}}
\newcommand{\bv}{\mathbf{v}}
\newcommand{\bp}{\mathbf{p}}
\newcommand{\bq}{\mathbf{q}}
\newcommand{\id}{\mathrm{id}}
\newcommand{\Ass}{\mathrm{Ass}}
\DeclareMathOperator{\disc}{disc}
\DeclareMathOperator{\Sym}{Sym}
\DeclareMathOperator{\tr}{tr}

\begin{document}

\begin{frontmatter}

\title{A solution to the quantum Yang--Baxter equation associated with a universal two-valued algebraic group}

\author[smi,simc,yar]{Victor M. Buchstaber}
\ead{buchstab@mi-ras.ru}

\affiliation[smi]{organization={Steklov Mathematical Institute of the Russian Academy of Sciences},
            country={Russia}}
\affiliation[simc]{organization={Steklov International Mathematical Center},
            city={Moscow},
            country={Russia}}
\affiliation[yar]{organization={Centre of Integrable Systems, P.G.~Demidov Yaroslavl State University},
            city={Yaroslavl},
            country={Russia}}

\begin{abstract}
We construct spectral-parameter solutions of the quantum Yang--Baxter equation from the associativity condition of the universal symmetric $2$-algebraic two-valued group, whose multiplication law depends on four coefficients $k_2,k_4,k_6,k_8$ subject to the single relation $4k_8-k_4^2+k_2k_6=0$ -- which coincides with the three-dimensional WDVV equation for Dubrovin's normalised potential. Every three-dimensional Frobenius manifold with a cyclic vector field carries such a spectral two-valued group, whose finite branch points are the eigenvalues of the multiplication operator. For an arbitrary unital algebra we show that a two-parameter spectral $R$-operator built from multiplication and unit has Yang--Baxter defect proportional to an explicit tensor built from the associator; for tangent Frobenius algebras this yields a regular, unitary $9\times9$ $R$-matrix satisfying the quantum Yang--Baxter equation exactly at WDVV points, with spectral parameter the Abel coordinate of a bielliptic curve whose elliptic quotient is the spectral curve, together with an associated integrable chain. The $A_3$ Frobenius manifold is worked out explicitly.
\end{abstract}

\begin{keyword}
Two-valued groups \sep WDVV equations \sep Frobenius manifolds \sep Quantum Yang--Baxter equation \sep Elliptic curves \sep Hydrodynamic-type systems \sep Bielliptic curves
\MSC[2020] 16T25 \sep 53D45 \sep 20N20 \sep 14H52 \sep 37K10 \sep 35L65
\end{keyword}

\end{frontmatter}

\section{Introduction}\label{sec:intro}

The theory of multivalued groups originated in algebraic topology. Two-valued formal groups appeared in the work of Buchstaber and Novikov on characteristic classes and Adams operations in complex cobordism~\cite{BN1971}, and their algebraic theory was developed in~\cite{B1975}; see~\cite{B2006} for a survey of $n$-valued groups and their applications. A basic example is provided by an elliptic curve $E$: the addition law of $E$ induces on the quotient $E/\{p\sim-p\}\simeq\PP^1$ a two-valued multiplication, the product of the classes of $p$ and $q$ being the unordered pair of the classes of $p+q$ and $p-q$. In~\cite{B1990} this law was written as the zero set of an explicit symmetric polynomial in three variables. More recently, algebraic two-valued group structures on $\PP^1$ and on $\C$ were studied systematically~\cite{BGR2024,BK2026,BKR2026}. In particular, all symmetric $2$-algebraic two-valued groups on $\C$ are specialisations of a universal law given by a polynomial $B$ with four coefficients $k_2,k_4,k_6,k_8$, and associativity is equivalent to the single quadratic relation $4k_8-k_4^2+k_2k_6=0$~\cite{BK2026}.

The Witten--Dijkgraaf--Verlinde--Verlinde (WDVV) equations arose in two-dimensional topological field theory and topological gravity~\cite{Witten1990,DVV1991} and were given a geometric formulation by Dubrovin in the theory of Frobenius manifolds~\cite{Dubrovin1992,Dubrovin1996}. Algebraically they express the associativity of a family of commutative Frobenius algebras. In dimension three they reduce to a single third-order partial differential equation for a function $f(x,y)$, which can be written as an integrable Hamiltonian system of hydrodynamic type~\cite{FM1996}.

The Yang--Baxter equation arose from two closely related problems of mathematical physics. In Yang's work on the one-dimensional quantum many-body problem with $\delta$-function interaction~\cite{Yang1967}, it is the condition that the scattering of many particles factorises consistently into a sequence of two-particle scattering processes. Baxter~\cite{Baxter1972} used the same condition, in the form of commuting transfer matrices, to solve the eight-vertex model of statistical mechanics.

Although the WDVV equations and the quantum Yang--Baxter equation (QYBE) come from different areas, their physical origins reveal a common principle: \emph{a complicated process may be decomposed into elementary operations in several different ways, and the final result must be independent of the chosen decomposition.} For the WDVV equations the elementary operations are multiplications in a Frobenius algebra; for the QYBE they are two-particle scatterings. The aim of this paper is to show that, in dimension three, both equations are governed by one algebraic object: the associativity equation of the universal symmetric $2$-algebraic two-valued group. Along the way this equation acquires interpretations in terms of elliptic and bielliptic integrals, Frobenius algebras and manifolds, spectral curves and commuting operators, and integrable systems of hydrodynamic type (Fig.~\ref{fig:scheme}).

Our main results are as follows.
\begin{enumerate}
\item \emph{Elliptic realisation} (Section~\ref{sec:elliptic}). For a monic cubic $P(\lambda)=\lambda^3-s_1\lambda^2+s_2\lambda-s_3$, the two-valued addition law of the curve $w^2=P(\lambda)$ is given by the polynomial $(\sigma_2-s_2)^2-4(\sigma_1-s_1)(\sigma_3-s_3)$ in the elementary symmetric functions of the three arguments. After inversion of the coordinate it becomes the universal law with $k=(4s_1,-2s_2,4s_3,s_2^2-4s_1s_3)$, and every associative universal law arises in this way. The logarithm of the law is the square of an elliptic integral, and the exponential is expressed through the Weierstrass function $\wp$.
\item \emph{WDVV and associativity} (Section~\ref{sec:wdvv}). For Dubrovin's normalised three-dimensional potential, the identification $k=4(f_{xxx},f_{xxy},f_{xyy},f_{yyy})$ turns the WDVV equation into the associativity relation of the universal two-valued group, and into the commutativity of the operators of multiplication by $e_2$ and $e_3$ (Theorem~\ref{thm:main-correspondence}; see also~\cite{BKR2026}). A cyclic vector field on a three-dimensional Frobenius manifold defines a \emph{spectral structure}: a family of elliptic two-valued groups whose finite branch points are the eigenvalues of the multiplication operator.
\item \emph{Nonlinear dynamics of the coefficients} (Section~\ref{sec:hydro}). Along a solution of the WDVV equation the coefficients $k_j$ satisfy a system of conservation laws of hydrodynamic type; its characteristic velocities are the finite branch points of the spectral two-valued law. Conversely, a solution of this system determines the WDVV potential up to a quadratic polynomial, and the potential is recovered from the logarithms of the two-valued groups (a local inverse spectral theorem).
\item \emph{The quantum Yang--Baxter equation} (Section~\ref{sec:qybe}). For an arbitrary unital algebra $\cA$ we consider the two-parameter family of operators
$R(U,V)(\xi\otimes\eta)=\alpha(U-V)\,1\otimes\xi\eta+\beta(U-V)\,\xi\eta\otimes1-(\alpha U-\beta V)\,\eta\otimes\xi$
and prove that its Yang--Baxter defect is $(\alpha+\beta)(U-V)(U-W)(V-W)$ times an explicit tensor built from the associator of $\cA$ (Theorem~\ref{thm:defect}). For the three-dimensional tangent algebras of a Frobenius manifold this gives an explicit $9\times9$ matrix $R^{\mathrm{bi}}_{p;\alpha,\beta}(U,V)$ which is regular and unitary and which, for $\alpha+\beta\neq0$, satisfies the QYBE if and only if the WDVV equation holds at $p$ (Theorem~\ref{thm:main}). The spectral parameters are provided by the Abel coordinate $U=\sqrt{\LL_p(X)}$ of the two-valued group, a bielliptic integral on a genus-two curve whose elliptic quotient is the spectral curve. The regularity of the $R$-matrix yields an integrable nearest-neighbour chain in which particles hop and fuse according to the multiplication table of the Frobenius algebra (Proposition~\ref{prop:hamiltonian}).
\item \emph{Example} (Section~\ref{sec:A3}). For Dubrovin's $A_3$ potential the flat coordinates are, up to constant factors, the Weierstrass invariants of the spectral curve and the coupling constants of the integrable chain.
\end{enumerate}

Several approaches to spectral-parameter dependent solutions of the QYBE compose spectral parameters by means of one-valued formal group laws; see, for example, the lattice models of~\cite{BFHTW2023} and the references therein. One-valued formal groups have also been useful in other problems related to quantum groups and Yang--Baxter-type structures. However, the applicability of these approaches is limited when one works with algebraic constructions that employ a global rational affine spectral coordinate: the one-dimensional group laws that are rational in such a coordinate form only a two-parameter family, the formal group law of the two-parameter Todd genus~\cite{Hirzebruch1966}. (This is the classification of one-dimensional rational formal group laws over a field: the first proof is due to~\cite{Bismuth1976}, with a short proof for fields of characteristic zero in~\cite{ColemanMcGuinness1991}; the two-parameter genus itself is a homogeneous version of Hirzebruch's $\chi_y$-genus, a construction going back to~\cite{Krichever1974}.) The situation changes when one turns to two-valued groups. Our construction relies on the four-parameter universal algebraic two-valued group (with one relation) which, as we show, admits a natural realisation within the theory of WDVV equations.

The paper is organised as follows. Section~\ref{sec:twovalued} recalls two-valued groups and the universal symmetric law. Section~\ref{sec:elliptic} gives its elliptic realisation, logarithm and exponential. Section~\ref{sec:wdvv} treats Frobenius manifolds, the three-dimensional WDVV equation and the spectral structure. Section~\ref{sec:hydro} is devoted to the hydrodynamic form of the WDVV equation and to reconstruction of the potential. Section~\ref{sec:qybe} contains the construction of the $R$-matrix, the main theorem and the associated integrable chain, Section~\ref{sec:A3} treats the example of the $A_3$ Frobenius manifold, and Section~\ref{sec:discussion} contains concluding remarks. The expansion used in the proof of Theorem~\ref{thm:defect} is given in \ref{app:expansion}. All polynomial identities stated below have also been checked with the computer algebra system SymPy; the scripts are provided as supplementary material.

\begin{figure}[t]
\centering
\begin{tikzpicture}[
  box/.style={draw, rounded corners=2pt, align=center, inner sep=4pt, font=\small, text width=4.6cm},
  arr/.style={{Latex[length=2mm]}-{Latex[length=2mm]}, thick}]
\node[box] (assoc) {Associativity of the universal\\ two-valued group\\ $4k_8-k_4^2+k_2k_6=0$};
\node[box, right=1.6cm of assoc] (wdvv) {3D WDVV equation\\ $f_{xxy}^2=f_{xxx}f_{xyy}+f_{yyy}$};
\node[box, below=1.1cm of assoc] (ell) {Two-valued law of the\\ spectral curve $w^2=\det(\lambda-L_{e_2})$};
\node[box, below=1.1cm of wdvv] (comm) {Commuting multiplication\\ operators $[L_{e_2},L_{e_3}]=0$};
\node[box, below=1.1cm of ell] (hydro) {Conservation laws for $k_j$;\\ velocities $=$ branch points};
\node[box, below=1.1cm of comm] (qybe) {QYBE for $R^{\mathrm{bi}}_{p;\alpha,\beta}(U,V)$,\\ $\alpha+\beta\neq0$};
\draw[arr] (assoc) -- node[above,font=\footnotesize]{Thm~\ref{thm:main-correspondence}} (wdvv);
\draw[arr] (assoc) -- node[left,font=\footnotesize]{Prop.~\ref{prop:inversion}} (ell);
\draw[arr] (wdvv) -- node[right,font=\footnotesize]{Lemma~\ref{lem:MN}} (comm);
\draw[arr] (wdvv.south west) -- node[above,sloped,font=\footnotesize]{Prop.~\ref{prop:spectral}} (ell.north east);
\draw[arr] (ell) -- node[left,font=\footnotesize]{Sect.~\ref{sec:hydro}} (hydro);
\draw[arr] (comm) -- node[right,font=\footnotesize]{Thm~\ref{thm:main}} (qybe);
\end{tikzpicture}
\caption{The equivalent forms of the associativity condition discussed in this paper. All statements refer to a point $p$ of a three-dimensional Frobenius manifold with the normalised potential~\eqref{eq:potential}, and $(k_2,k_4,k_6,k_8)=4(f_{xxx},f_{xxy},f_{xyy},f_{yyy})$.}
\label{fig:scheme}
\end{figure}
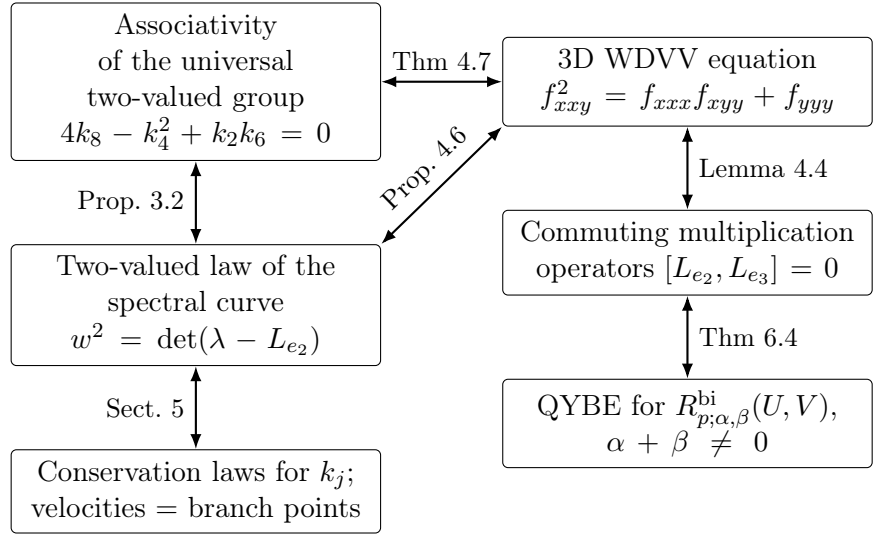

\section{Two-valued algebraic groups}\label{sec:twovalued}

\subsection{Definitions}

For a set $\mathcal X$ we denote by $\Sym^n\mathcal X$ the set of unordered $n$-tuples (multisets) $[x_1,\dots,x_n]$ of elements of $\mathcal X$, with multiplicities retained.

\begin{definition}[\cite{B1975,B2006}]\label{def:twovalued}
A \emph{two-valued multiplication} on $\mathcal X$ is a map $\mathcal X\times\mathcal X\to\Sym^2\mathcal X$, $(x,y)\mapsto x*y=[z_1,z_2]$. It is extended to multisets by $[z_1,z_2]*y=[z_1*y,\,z_2*y]\in\Sym^4\mathcal X$ (union of multisets), and similarly on the other side. The multiplication is \emph{associative} if $(x*y)*z=x*(y*z)$ in $\Sym^4\mathcal X$ for all $x,y,z$. An element $e$ is a \emph{unit} if $e*x=x*e=[x,x]$, and a map $\mathrm{inv}\colon\mathcal X\to\mathcal X$ is an \emph{inverse} if $e\in\mathrm{inv}(x)*x$ and $e\in x*\mathrm{inv}(x)$ for all $x$. A \emph{two-valued group} is a set with an associative two-valued multiplication, a unit and an inverse; it is \emph{commutative} if $x*y=y*x$.
\end{definition}

A two-valued multiplication on $\C$ is \emph{algebraic} if it is described by a polynomial
\begin{equation}\label{eq:P2}
P_2(Z;X,Y)=F_0(X,Y)Z^2+F_1(X,Y)Z+F_2(X,Y)
\end{equation}
whose roots in $Z$ are the two values of $X*Y$.

\begin{definition}\label{def:symmetric}
A \emph{symmetric $2$-algebraic two-valued group} on $\C$ is an algebraic two-valued group with unit $0$ and inverse $\mathrm{inv}(X)=X$, whose multiplication is given by a polynomial~\eqref{eq:P2} that is symmetric in the three variables $X,Y,Z$ and has degree at most $2$ in each of them.
\end{definition}

In terms of $P_2$ the axioms read as follows. The unit axiom $X*0=[X,X]$ is equivalent to $P_2(Z;X,0)=F_0(X,0)(Z-X)^2$ with $F_0(X,0)\not\equiv0$. The inverse axiom $0\in X*X$ is equivalent to $F_2(X,X)=0$. Commutativity $X*Y=Y*X$ is equivalent to $P_2(Z;X,Y)=P_2(Z;Y,X)$ and is automatic for symmetric $P_2$. If $Z_1,Z_2$ are the roots of $P_2(Z;X,Y)$, associativity is the equality of the four-valued multisets $(X*Y)*Z$ and $X*(Y*Z)$; equivalently, the two monic quartic polynomials in an auxiliary variable $T$ whose roots are these four values coincide identically in $X,Y,Z$. This is the polynomial form of the associativity axiom. The quartics can be computed as resultants: up to factors independent of $T$ they are $\operatorname{Res}_W\bigl(P_2(W;X,Y),P_2(T;W,Z)\bigr)$ and $\operatorname{Res}_W\bigl(P_2(W;Y,Z),P_2(T;X,W)\bigr)$.

\subsection{The elliptic two-valued group}

Let $E$ be a smooth projective curve of genus one with a chosen point $O$; this choice gives an addition law on $E$ with neutral element $O$. Let
\[
\pi\colon E\to E/\{p\sim-p\}\simeq\PP^1
\]
be the quotient map. For $p,q\in E$ the addition induces the two-valued product
\begin{equation}\label{eq:elliptic-law}
\pi(p)*\pi(q)=[\pi(p+q),\pi(p-q)],
\end{equation}
which is well defined because the right-hand side does not change when $p$ or $q$ is replaced by its negative. Choose an affine coordinate on $\PP^1$ in which $\pi(O)=0$. In~\cite{B1990} the law~\eqref{eq:elliptic-law} is represented, for a vector of parameters $\mathbf a=(a_1,a_2,a_3)\in\C^3$, by the polynomial
\begin{multline}\label{eq:Ba}
B_{\mathbf a}(X,Y,Z)=(X+Y+Z-a_2XYZ)^2\\
-4(1+a_3XYZ)(XY+XZ+YZ+a_1XYZ).
\end{multline}
With $\sigma_1=X+Y+Z$, $\sigma_2=XY+XZ+YZ$, $\sigma_3=XYZ$ this reads
\begin{equation}\label{eq:Ba-sigma}
B_{\mathbf a}=\sigma_1^2-4\sigma_2-4a_1\sigma_3-2a_2\sigma_1\sigma_3-4a_3\sigma_2\sigma_3+(a_2^2-4a_1a_3)\sigma_3^2.
\end{equation}

\subsection{The universal symmetric \texorpdfstring{$2$}{2}-algebraic two-valued group}

For $k=(k_2,k_4,k_6,k_8)\in\C^4$ put
\begin{equation}\label{eq:B}
B_k(X,Y,Z)=\sigma_1^2-4\sigma_2+k_2\sigma_3+k_4\sigma_1\sigma_3+k_6\sigma_2\sigma_3+k_8\sigma_3^2 .
\end{equation}
The indices indicate the natural grading: if $X,Y,Z$ have weight $-2$, then $k_{2j}$ has weight $2j$ and $B_k$ is homogeneous.

\begin{proposition}\label{prop:normal-form}
Let $P_2(Z;X,Y)$ be symmetric in $X,Y,Z$, of degree at most $2$ in each variable, and satisfy the unit axiom. Then $P_2=c\,B_k$ for some $c\in\C^\times$ and $k\in\C^4$. Conversely, for every $k\in\C^4$ the polynomial $B_k$ satisfies the unit, inverse and commutativity axioms.
\end{proposition}

\begin{proof}
A symmetric polynomial of degree at most $2$ in each variable is a linear combination of the ten monomial symmetric functions $m_\mu$, $\mu\subseteq(2,2,2)$. Setting $Y=0$ kills the $m_\mu$ with three parts and leaves
\[
\begin{aligned}
P_2(Z;X,0)={}&c_0+c_1(X+Z)+c_{11}XZ+c_{2}(X^2+Z^2)\\
&+c_{21}(X^2Z+XZ^2)+c_{22}X^2Z^2 .
\end{aligned}
\]
The unit axiom requires this to vanish at $Z=X$ identically in $X$, which gives $c_0=c_1=c_{21}=c_{22}=0$ and $c_{11}=-2c_2$; then $P_2(Z;X,0)=c_2(Z-X)^2$ and $c_2\ne0$. Hence
\[
P_2=c_2(m_2-2m_{11})+c_{111}m_{111}+c_{211}m_{211}+c_{221}m_{221}+c_{222}m_{222},
\]
and it remains to note that $m_2-2m_{11}=\sigma_1^2-4\sigma_2$, $m_{111}=\sigma_3$, $m_{211}=\sigma_1\sigma_3$, $m_{221}=\sigma_2\sigma_3$, $m_{222}=\sigma_3^2$. Conversely, $B_k(X,0,Z)=(Z-X)^2$ and $B_k(X,X,0)=(2X)^2-4X^2=0$.
\end{proof}

\begin{theorem}[\cite{BK2026}]\label{thm:universal}
The two-valued multiplication defined by $B_k$ is associative if and only if
\begin{equation}\label{eq:assoc}
4k_8-k_4^2+k_2k_6=0 .
\end{equation}
Hence the coefficient ring of the universal symmetric $2$-algebraic two-valued group is $\Z[k_2,k_4,k_6,k_8]/(4k_8-k_4^2+k_2k_6)$.
\end{theorem}

A conceptual proof of sufficiency is given in Corollary~\ref{cor:sufficiency} below. Necessity can be checked directly: if $Q_1,Q_2$ are the two resultants above (with $P_2=B_k$) and $l_1,l_2$ their leading coefficients in $T$, then
\[
l_2Q_1-l_1Q_2=(4k_8-k_4^2+k_2k_6)\,XYZ(X-Z)\,G_k(X,Y,Z,T),
\]
where $G_k$ is a polynomial in $X,Y,Z,T$ some of whose coefficients are nonzero numerical constants, independent of $k$. Therefore, if~\eqref{eq:assoc} fails, the two quartics are different.

\begin{corollary}\label{cor:elliptic-special}
The elliptic law~\eqref{eq:Ba} is the specialisation of the universal law given by
\[
k_2=-4a_1,\qquad k_4=-2a_2,\qquad k_6=-4a_3,\qquad k_8=a_2^2-4a_1a_3,
\]
and these values satisfy~\eqref{eq:assoc}.
\end{corollary}

\section{Elliptic realisation, logarithm and exponential}\label{sec:elliptic}

\subsection{The two-valued law of a cubic}

For $s=(s_1,s_2,s_3)\in\C^3$ let
\begin{equation}\label{eq:Ps}
P_s(\lambda)=\lambda^3-s_1\lambda^2+s_2\lambda-s_3 ,
\end{equation}
and for $\lambda_1,\lambda_2,\lambda_3\in\C$ put
\begin{equation}\label{eq:Ds-def}
D_s(\lambda_1,\lambda_2,\lambda_3)=\disc_\lambda\bigl[P_s(\lambda)-(\lambda-\lambda_1)(\lambda-\lambda_2)(\lambda-\lambda_3)\bigr],
\end{equation}
where the discriminant of $a\lambda^2+b\lambda+c$ is $b^2-4ac$ (also when $a=0$). Let $\sigma_i(\lambda)$ denote the elementary symmetric functions of $\lambda_1,\lambda_2,\lambda_3$.

\begin{proposition}\label{prop:Ds}
\begin{enumerate}
\item[\textup{(i)}] $D_s=(\sigma_2(\lambda)-s_2)^2-4(\sigma_1(\lambda)-s_1)(\sigma_3(\lambda)-s_3)$; it is symmetric and of degree at most $2$ in each $\lambda_i$.
\item[\textup{(ii)}] As a polynomial in $\lambda_3$, $D_s$ has leading coefficient $(\lambda_1-\lambda_2)^2$ and discriminant
\[
\disc_{\lambda_3}D_s(\lambda_1,\lambda_2,\lambda_3)=2^4P_s(\lambda_1)P_s(\lambda_2).
\]
\item[\textup{(iii)}] Suppose $P_s$ has three distinct roots, and let $E_s\colon w^2=P_s(\lambda)$ be the corresponding elliptic curve with neutral element $O$ at infinity. If $\bp,\bq\in E_s$ and $\lambda(\bp)\ne\lambda(\bq)$, then the two roots of $D_s(\lambda(\bp),\lambda(\bq),\,\cdot\,)$ are $\lambda(\bp+\bq)$ and $\lambda(\bp-\bq)$.
\end{enumerate}
\end{proposition}

\begin{proof}
(i) Since $P_s(\lambda)-\prod_i(\lambda-\lambda_i)=(\sigma_1(\lambda)-s_1)\lambda^2+(s_2-\sigma_2(\lambda))\lambda+(\sigma_3(\lambda)-s_3)$, the formula follows. (ii) is a direct computation.
(iii) Assume first that $\bp,\bq,\bp\pm\bq$ are distinct from $O$ and that $\lambda(\bp+\bq)\neq\lambda(\bp-\bq)$. The line $w=m\lambda+n$ through $\bp$ and $\bq$ meets $E_s$ in the third point $-(\bp+\bq)$, whose $\lambda$-coordinate is $\lambda(\bp+\bq)$. Hence $P_s(\lambda)-(m\lambda+n)^2=(\lambda-\lambda(\bp))(\lambda-\lambda(\bq))(\lambda-\lambda(\bp+\bq))$, so that $P_s(\lambda)-(\lambda-\lambda(\bp))(\lambda-\lambda(\bq))(\lambda-\lambda(\bp+\bq))$ is a square and its discriminant vanishes. Replacing $\bq$ by $-\bq$ gives the root $\lambda(\bp-\bq)$. Since the leading coefficient $(\lambda(\bp)-\lambda(\bq))^2$ is nonzero, these are all the roots. The remaining cases follow by continuity.
\end{proof}

By (ii), for fixed $\lambda_1$ with $P_s(\lambda_1)\ne0$, the two-valued function $\lambda_2\mapsto\lambda_1*\lambda_2$ is branched exactly at the roots of $P_s$, that is, at the $\lambda$-coordinates of the nonzero points of order two; when $\lambda_2=\lambda_1$ one of its values is the neutral element $\lambda=\infty$.

\begin{proposition}\label{prop:inversion}
For every $s\in\C^3$,
\begin{equation}\label{eq:inversion}
\begin{gathered}
(XYZ)^2\,D_s\bigl(X^{-1},Y^{-1},Z^{-1}\bigr)=B_{k(s)}(X,Y,Z),\\
k(s)=\bigl(4s_1,\,-2s_2,\,4s_3,\,s_2^2-4s_1s_3\bigr).
\end{gathered}
\end{equation}
The map $s\mapsto k(s)$ is a bijection of $\C^3$ onto the quadric $\mathcal Q=\{4k_8-k_4^2+k_2k_6=0\}$, with inverse $s=(k_2/4,-k_4/2,k_6/4)$.
\end{proposition}

\begin{proof}
Writing $\tau_i$ for the elementary symmetric functions of $X,Y,Z$, we have $\sigma_1(\lambda)=\tau_2/\tau_3$, $\sigma_2(\lambda)=\tau_1/\tau_3$, $\sigma_3(\lambda)=1/\tau_3$ for $\lambda_1=X^{-1}$, etc. Multiplying Proposition~\ref{prop:Ds}(i) by $\tau_3^2$ gives
\[
\begin{aligned}
&(\tau_1-s_2\tau_3)^2-4(\tau_2-s_1\tau_3)(1-s_3\tau_3)\\
&\qquad=\tau_1^2-4\tau_2+4s_1\tau_3-2s_2\tau_1\tau_3+4s_3\tau_2\tau_3+(s_2^2-4s_1s_3)\tau_3^2 ,
\end{aligned}
\]
which is~\eqref{eq:inversion}. Finally, $4(s_2^2-4s_1s_3)-4s_2^2+16s_1s_3=0$.
\end{proof}

\begin{corollary}\label{cor:sufficiency}
If $k\in\mathcal Q$, the law $B_k$ is associative. Thus the associative universal laws are exactly the two-valued addition laws of the curves $w^2=P_s(\lambda)$, written in the coordinate $X=1/\lambda$.
\end{corollary}

\begin{proof}
Let $k=k(s)$. If $P_s$ has distinct roots, Propositions~\ref{prop:Ds}(iii) and~\ref{prop:inversion} show that $B_k$ describes the law~\eqref{eq:elliptic-law} of $E_s$ in the coordinate $X=1/\lambda(\bp)$, with unit $X=0$ corresponding to $O$. For $X_i=X(\bp_i)$ both $(X_1*X_2)*X_3$ and $X_1*(X_2*X_3)$ are equal to the multiset of the four values $X(\bp_1+\epsilon_2\bp_2+\epsilon_3\bp_3)$, $\epsilon_2,\epsilon_3\in\{\pm1\}$, so the law is associative. The coefficients of the two monic quartics of Section~\ref{sec:twovalued} are rational functions of $(X_1,X_2,X_3)$ and polynomial in $s$; their equality on the dense set of $s$ with $\disc P_s\ne0$ implies it for all $s$.
\end{proof}

In particular, the polynomial two-valued law extends algebraically across the discriminant locus $\disc P_s=0$, although the elliptic curve itself degenerates there.

\subsection{Logarithm and exponential}

Let $k\in\mathcal Q$ and put
\begin{equation}\label{eq:g2g3}
g_2=2k_4+\frac{k_2^2}{12},\qquad g_3=k_6+\frac{k_2k_4}{6}+\frac{k_2^3}{216}.
\end{equation}
Let $\wp(u)=\wp(u;g_2,g_3)$ be the Weierstrass function; when $g_2^3=27g_3^2$ we understand $\wp$ as its Laurent series $u^{-2}+\frac{g_2}{20}u^2+\frac{g_3}{28}u^4+\dots$ at $u=0$, whose coefficients are polynomials in $g_2,g_3$.

\begin{proposition}\label{prop:log}
Let $k\in\mathcal Q$ and $s=(k_2/4,-k_4/2,k_6/4)$.
\begin{enumerate}
\item[\textup{(i)}] The substitution $\lambda=\wp(u)+k_2/12$, $w=\wp'(u)/2$ transforms $\wp'^2=4\wp^3-g_2\wp-g_3$ into $w^2=P_s(\lambda)$.
\item[\textup{(ii)}] The \emph{logarithm}
\begin{equation}\label{eq:log}
\LL(X)=\left(\int_0^{\sqrt X}\frac{dt}{\sqrt{1-\frac{k_2}{4}t^2-\frac{k_4}{2}t^4-\frac{k_6}{4}t^6}}\right)^2=X+O(X^2)
\end{equation}
and the \emph{exponential}
\begin{equation}\label{eq:exp}
\EE(S)=\frac{1}{\wp(\sqrt S)+\frac{k_2}{12}}=S+O(S^2)
\end{equation}
are power series with coefficients in $\mathbb Q[k_2,k_4,k_6]$, and they are mutually inverse: $\LL(\EE(S))=S$.
\item[\textup{(iii)}] Put $U(X)=\sqrt{\LL(X)}=\sqrt X\,(1+O(X))$. Then
\begin{equation}\label{eq:law-log}
X*Y=\Bigl[\EE\bigl((U(X)+U(Y))^2\bigr),\ \EE\bigl((U(X)-U(Y))^2\bigr)\Bigr],
\end{equation}
that is, $U(Z_\pm)=U(X)\pm U(Y)$ for the two values $Z_\pm$ of $X*Y$. Equivalently, $B_k\bigl(\EE(u^2),\EE(v^2),\EE((u\pm v)^2)\bigr)=0$ identically in $u,v$.
\end{enumerate}
\end{proposition}

\begin{proof}
(i) Substituting $\lambda=\wp+k_2/12$ into $P_s(\lambda)=\lambda^3-\frac{k_2}{4}\lambda^2-\frac{k_4}{2}\lambda-\frac{k_6}{4}$ removes the quadratic term and gives $\wp^3-\frac{g_2}{4}\wp-\frac{g_3}{4}=\frac14\wp'^2=w^2$ with $g_2,g_3$ as in~\eqref{eq:g2g3}.
(ii) Since $\wp$ is even, $\EE$ is a series in $S=u^2$. By (i), $du=d\lambda/(2w)$. For $\lambda=1/X$ and $X=t^2$ we get $d\lambda/(2w)=\mp dt/\sqrt{1-s_1t^2+s_2t^4-s_3t^6}$, and $1-s_1t^2+s_2t^4-s_3t^6=1-\frac{k_2}{4}t^2-\frac{k_4}{2}t^4-\frac{k_6}{4}t^6$. Since $u=0$ corresponds to $X=0$ and only $u^2$ enters $\EE$, we may fix the sign of $u$ so that $u=U(\EE(u^2))$, i.e.\ $\LL(\EE(S))=S$.
(iii) If $\disc P_s\ne0$, the map $u\mapsto(\lambda,w)$ is the uniformisation $\C\to E_s$, a group homomorphism. Hence $1/\EE((u\pm v)^2)$ are the $\lambda$-coordinates of the sum and difference of the points with parameters $u$ and $v$, and the claim follows from Propositions~\ref{prop:Ds}(iii) and~\ref{prop:inversion}. The coefficients of the series $B_k(\EE(u^2),\EE(v^2),\EE((u\pm v)^2))$ are polynomials in $k_2,k_4,k_6$ that vanish on a dense set, hence identically.
\end{proof}

Thus the logarithm and the exponential determine the two-valued law completely. The coefficient $k_8$ does not enter $\LL$; it is recovered from the associativity relation~\eqref{eq:assoc}.

\section{Frobenius manifolds, the WDVV equation and the spectral structure}\label{sec:wdvv}

\subsection{Frobenius algebras and Frobenius manifolds}

\begin{definition}
A \emph{Frobenius algebra} over $\C$ is a finite-dimensional unital associative algebra $A$ together with a linear functional $\theta\colon A\to\C$ such that the bilinear form $\eta_A(a,b)=\theta(ab)$ is nondegenerate, i.e.\ $\eta_A(a,b)=0$ for all $b\in A$ implies $a=0$. The form is then invariant: $\eta_A(ab,c)=\eta_A(a,bc)$. The Frobenius algebra is \emph{symmetric} if $\eta_A(a,b)=\eta_A(b,a)$ for all $a,b$; every commutative Frobenius algebra is symmetric.
\end{definition}

We use the ungraded definition of a Frobenius manifold, without the Euler vector field axiom~\cite{Dubrovin1996}.

\begin{definition}
A \emph{Frobenius manifold} is a complex manifold $\mathcal M$ equipped with
\begin{itemize}
\item a holomorphic, commutative and associative product $u\circ v$ of vector fields with a unit vector field $e$;
\item a flat nondegenerate holomorphic metric $\eta$ invariant under the product, $\eta(u\circ v,w)=\eta(u,v\circ w)$;
\end{itemize}
such that, for the Levi-Civita connection $\nabla$ of $\eta$, $\nabla e=0$ and the tensor $\nabla c$ is totally symmetric, where $c(u,v,w)=\eta(u\circ v,w)$. Total symmetry means that $(\nabla_zc)(u,v,w)$ is invariant under all permutations of the vector fields $z,u,v,w$.
\end{definition}

Every tangent space $T_p\mathcal M$ is a commutative Frobenius algebra with unit $e_p$, form $\eta_p$ and functional $\theta_p(a)=\eta_p(a,e_p)$.

\subsection{The WDVV equations}

In flat coordinates $t^1,\dots,t^n$ with $e=\partial/\partial t^1$ there is locally a function $F(t)$, the potential, such that $c(\partial_i,\partial_j,\partial_k)=F_{ijk}:=\partial^3F/\partial t^i\partial t^j\partial t^k$ and $\eta_{ij}=F_{1ij}$ is a constant nondegenerate symmetric matrix. With $(\eta^{ij})=(\eta_{ij})^{-1}$ the multiplication of the basis vectors $e_i=\partial/\partial t^i$ is
\[
e_i\circ e_j=c_{ij}^{\,k}e_k,\qquad c_{ij}^{\,k}=\eta^{kl}F_{ijl},
\]
with summation over repeated indices. Then $e_1$ is the unit and the multiplication is commutative; its associativity $(e_i\circ e_j)\circ e_k=e_i\circ(e_j\circ e_k)$ is equivalent to the WDVV equations~\cite{Witten1990,DVV1991,Dubrovin1992}
\begin{equation}\label{eq:WDVV-general}
F_{ijm}\,\eta^{mn}\,F_{nkl}=F_{ikm}\,\eta^{mn}\,F_{njl}.
\end{equation}

\subsection{The three-dimensional Dubrovin algebra}

Consider Dubrovin's normalised three-dimensional potential~\cite{Dubrovin1996}
\begin{equation}\label{eq:potential}
F(t^1,t^2,t^3)=\tfrac12(t^1)^2t^3+\tfrac12t^1(t^2)^2+f(x,y),\qquad x=t^2,\ y=t^3,
\end{equation}
so that $\eta_{13}=\eta_{22}=1$ and all other $\eta_{ij}=0$, and put
\begin{equation}\label{eq:abcd}
b=f_{xxx},\qquad a=f_{xxy},\qquad c=f_{xyy},\qquad d=f_{yyy}.
\end{equation}
The tangent multiplication is
\begin{equation}\label{eq:table}
e_1=1,\qquad e_2^2=ae_1+be_2+e_3,\qquad e_2e_3=ce_1+ae_2,\qquad e_3^2=de_1+ce_2 .
\end{equation}
For arbitrary $(a,b,c,d)\in\C^4$ the table~\eqref{eq:table} defines a commutative unital three-dimensional algebra $A_{(a,b,c,d)}$, and the form $\eta$ is invariant for all values of the parameters. We write $(\xi_1,\xi_2,\xi_3)=(\xi_1\xi_2)\xi_3-\xi_1(\xi_2\xi_3)$ for the associator and
\begin{equation}\label{eq:Phi}
\Phi=a^2-bc-d .
\end{equation}

\begin{lemma}\label{lem:assoc}
In $A_{(a,b,c,d)}$ the nonzero associators of basis vectors are
\[
(e_2,e_2,e_3)=-(e_3,e_2,e_2)=-\Phi e_1,\qquad (e_2,e_3,e_3)=-(e_3,e_3,e_2)=\Phi e_2 .
\]
Consequently $A_{(a,b,c,d)}$ is associative if and only if $\Phi=0$, which is the three-dimensional WDVV equation
\begin{equation}\label{eq:WDVV3}
f_{xxy}^2=f_{xxx}f_{xyy}+f_{yyy}.
\end{equation}
\end{lemma}

\begin{proof}
For instance, $(e_2e_2)e_3=a e_3+b(ce_1+ae_2)+(de_1+ce_2)$ and $e_2(e_2e_3)=c e_2+a(ae_1+be_2+e_3)$, whose difference is $(bc+d-a^2)e_1$. The other associators are computed in the same way; those involving $e_1$ vanish, and $(e_2,e_3,e_2)=(e_3,e_2,e_3)=0$ by commutativity.
\end{proof}

Thus the commutative four-parameter algebra becomes a three-dimensional associative Frobenius algebra $A_\cD$ precisely on the quadric
\[
\cD=\bigl\{(a,b,c,d)\in\C^4 \bigm| a^2-bc-d=0\bigr\}.
\]

\begin{lemma}\label{lem:MN}
Let $M=L_{e_2}$ and $N=L_{e_3}$ be the operators of multiplication by $e_2$ and $e_3$. In the basis $(e_1,e_2,e_3)$ (the $j$-th column contains the image of $e_j$)
\[
M=\begin{pmatrix}0&a&c\\ 1&b&a\\ 0&1&0\end{pmatrix},\qquad
N=\begin{pmatrix}0&c&d\\ 0&a&c\\ 1&0&0\end{pmatrix},\qquad
[M,N]=\Phi\begin{pmatrix}0&1&0\\ 0&0&-1\\ 0&0&0\end{pmatrix}.
\]
Hence the WDVV equation holds if and only if $[M,N]=0$. For all $(a,b,c,d)$ the vector $e_1$ is cyclic for $M$: the vectors $e_1$, $Me_1=e_2$, $M^2e_1=ae_1+be_2+e_3$ form a basis. The characteristic polynomial of $M$ is
\begin{equation}\label{eq:chiM}
\chi_M(t)=t^3-bt^2-2at-c ,
\end{equation}
and if $\Phi=0$ then $A_\cD\simeq\C[t]/(\chi_M(t))$ with $e_2\mapsto t$, and $N=M^2-bM-a$.
\end{lemma}

\begin{proof}
The matrices and the commutator are obtained directly from~\eqref{eq:table}. If $\Phi=0$, the homomorphism $\C[t]\to A_\cD$, $t\mapsto e_2$, is surjective because $e_3=e_2^2-be_2-ae_1$, and $\chi_M(e_2)=\chi_M(M)e_1=0$ by the Cayley--Hamilton theorem; comparing dimensions gives the isomorphism.
\end{proof}

\subsection{Physical origin: primary chiral algebras of two-dimensional topological field theory}\label{subsec:physical}

The algebra $A_{(a,b,c,d)}$ is not an ad hoc algebraic construction: in Dubrovin's formulation it is the tangent Frobenius algebra of a moduli space of two-dimensional topological field theories (TFT), and we recall this physical background, referring throughout to~\cite{Dubrovin1996} (see already~\cite{Dubrovin1992}).

In a 2D TFT the matter sector is a finite-dimensional space $A$ of primary observables, and the TFT axioms endow $A$ with a commutative Frobenius algebra structure -- the \emph{primary chiral algebra} of the theory -- whose invariant metric and multiplication tensor are genus-zero correlators
\[
\eta_{\alpha\beta}=\langle\phi_\alpha\phi_\beta\rangle_0,\qquad
c_{\alpha\beta\gamma}=\langle\phi_\alpha\phi_\beta\phi_\gamma\rangle_0
\]
in a basis $\phi_1=1,\phi_2,\dots,\phi_n$ (\cite[Thm.~2.1, Rem.~2.1]{Dubrovin1996}). Perturbing the theory by the integrated descendants of $\phi_2,\dots,\phi_n$ with coupling constants $t=(t^1,\dots,t^n)$ produces a family of primary chiral algebras $A_t$ satisfying the WDVV equations, i.e.\ the moduli space of couplings carries a Frobenius-manifold structure (\cite[Thm.~2.2]{Dubrovin1996}). Thus the tangent algebra $T_p\mathcal M$ at a point $p$ of a Frobenius manifold is, physically, the algebra of primary observables of the corresponding perturbed TFT.

For $n=3$ this is precisely~\eqref{eq:table}, and the four structure constants are the independent genus-zero three-point functions of the two nontrivial primaries,
\[
b=\langle\phi_2\phi_2\phi_2\rangle_{0,t},\quad
a=\langle\phi_2\phi_2\phi_3\rangle_{0,t},\quad
c=\langle\phi_2\phi_3\phi_3\rangle_{0,t},\quad
d=\langle\phi_3\phi_3\phi_3\rangle_{0,t},
\]
the remaining three-point functions being fixed by the normalisation of the unit and the metric. In this language Lemma~\ref{lem:assoc} says that the WDVV equation $\Phi=0$ is exactly associativity of the operator product of the perturbed primary observables.

\begin{remark}[Two physical realisations]\label{rmk:physical-realisations}
Two standard models realise~\eqref{eq:table} with $\Phi=0$ identically. \emph{(i)} The $A_3$ topological Landau--Ginzburg model, whose free energy is Dubrovin's potential~\eqref{eq:A3} below (\cite[App.~A, eq.~(A.7), Lecture~2]{Dubrovin1996}); the algebra is then the Jacobian ring $\C[z]/(W'(z))$ of the superpotential $W(z)=\tfrac14z^4-yz^2-xz$, and by~\eqref{eq:chiM} the characteristic polynomial of $L_{e_2}$ equals $W'(z)$, so the eigenvalues of $L_{e_2}$ -- the finite branch points of the spectral two-valued law, Proposition~\ref{prop:spectral}(iii) -- are exactly the critical points of the Landau--Ginzburg superpotential; when these are nondegenerate and pairwise distinct, Dubrovin's massive theory identifies differences of the corresponding critical values with soliton masses (\cite[Example~2.4]{Dubrovin1996}). \emph{(ii)} The $A$-model topological sigma-model with target $\C P^2$, on the small quantum cohomology locus, with $e_2^2=e_3$, $e_2e_3=e^{t^2}e_1$ (\cite[Example~2.2]{Dubrovin1996}), i.e.\ $a=b=d=0$, $c=e^{t^2}$ in the notation of~\eqref{eq:table}. Thus the three-dimensional Dubrovin algebra is not tied to a single microscopic realisation.
\end{remark}

We stress a delimitation that we maintain throughout the paper. What is established in~\cite{Dubrovin1996,Dubrovin1992} is the physical meaning of the algebra~\eqref{eq:table} itself, of its structure constants as genus-zero three-point functions, and of the WDVV equation as the associativity condition for the perturbed primary chiral algebra. The elliptic curve $w^2=P_D(\lambda)$ and the two-valued group of Sections~\ref{sec:twovalued}--\ref{sec:elliptic}, the spectral structure below, and the $R$-matrix~\eqref{eq:Rmatrix} of Section~\ref{sec:qybe} are additional structures introduced in the present work and in~\cite{BKR2026}; they should not be attributed to the standard physical interpretation of the Dubrovin algebra recalled above. In particular, an $R$-matrix constructed from~\eqref{eq:table} should not automatically be identified with a physical scattering matrix of the $A_3$ Landau--Ginzburg theory -- such an identification would require additional dynamical input beyond the WDVV/associativity structure used here. The safe physical statement is: the three-dimensional algebra entering our construction is the algebra of primary observables of a two-dimensional TFT, its coefficients are genus-zero three-point functions, and its associativity is the WDVV equation; this provides a direct and well-established physical origin for the algebraic input of the Yang--Baxter construction below, while keeping separate the new spectral geometry and the question of a scattering or vertex-model interpretation of the resulting $R$-matrix.

\subsection{Spectral structure}

Let $(\mathcal M,\circ,\eta,e)$ be a three-dimensional Frobenius manifold. A vector field $\bv$ is \emph{cyclic} on an open set $\mathcal M_0\subset\mathcal M$ if $e,\bv,\bv\circ\bv$ form a local frame there. Put
\[
P_\bv(\lambda;p)=\det\bigl(\lambda-L_\bv(p)\bigr)=\lambda^3-s_1(p)\lambda^2+s_2(p)\lambda-s_3(p),\qquad p\in\mathcal M_0,
\]
where $L_\bv(p)$ is the operator of multiplication by $\bv$ in $T_p\mathcal M$.

\begin{proposition}\label{prop:spectral}
Let $\bv$ be cyclic on $\mathcal M_0$ and $p\in\mathcal M_0$.
\begin{enumerate}
\item[\textup{(i)}] $T_p\mathcal M\simeq\C[t]/(P_\bv(t;p))$ as algebras.
\item[\textup{(ii)}] If $P_\bv(\,\cdot\,;p)$ has three distinct roots, the curve $E_p\colon w^2=P_\bv(\lambda;p)$ is an elliptic curve with neutral point at infinity, and its two-valued law $\lambda_1*_p\lambda_2=[\lambda(\bp+\bq),\lambda(\bp-\bq)]$ has the graph $D_\bv(\lambda_1,\lambda_2,\lambda_3)=0$, where
\[
\begin{aligned}
D_\bv&=\disc_\lambda\bigl[P_\bv(\lambda)-(\lambda-\lambda_1)(\lambda-\lambda_2)(\lambda-\lambda_3)\bigr]\\
&=(\sigma_2(\lambda)-s_2)^2-4(\sigma_1(\lambda)-s_1)(\sigma_3(\lambda)-s_3).
\end{aligned}
\]
\item[\textup{(iii)}] $\disc_{\lambda_3}D_\bv=2^4P_\bv(\lambda_1)P_\bv(\lambda_2)$; thus the finite branch points of the two-valued multiplication are precisely the eigenvalues of $L_\bv(p)$.
\item[\textup{(iv)}] In the coordinate $X=1/\lambda$ the law is the universal law $B_{k(p)}$ with
\[
k_2=4s_1,\qquad k_4=-2s_2,\qquad k_6=4s_3,\qquad k_8=s_2^2-4s_1s_3 ,
\]
which satisfies~\eqref{eq:assoc} identically and is defined for all $p\in\mathcal M_0$, including the discriminant locus.
\end{enumerate}
\end{proposition}

\begin{proof}
(i) is proved as in Lemma~\ref{lem:MN}, and (ii)--(iv) follow from Propositions~\ref{prop:Ds} and~\ref{prop:inversion} with $s=s(p)$.
\end{proof}

We call the family of two-valued groups $\{B_{k(p)}\}_{p\in\mathcal M_0}$ the \emph{spectral structure} of the cyclic three-dimensional Frobenius manifold $(\mathcal M,\bv)$. For the potential~\eqref{eq:potential} the vector field $\bv=e_2=\partial/\partial t^2$ is cyclic everywhere by Lemma~\ref{lem:MN}, and~\eqref{eq:chiM} gives $s_1=b$, $s_2=-2a$, $s_3=c$. Therefore the spectral construction gives
\begin{equation}\label{eq:spectral-k}
k_2=4b,\qquad k_4=4a,\qquad k_6=4c,\qquad k_8=4(a^2-bc).
\end{equation}

\subsection{The universal associativity correspondence}

\begin{theorem}\label{thm:main-correspondence}
For $(a,b,c,d)\in\C^4$ put
\begin{equation}\label{eq:k-abcd}
k_2=4b,\qquad k_4=4a,\qquad k_6=4c,\qquad k_8=4d .
\end{equation}
Then $4k_8-k_4^2+k_2k_6=-2^4(a^2-bc-d)$, and the following conditions are equivalent:
\begin{enumerate}
\item[\textup{(a)}] the WDVV equation $\Phi=0$ holds, i.e.\ $A_{(a,b,c,d)}$ is associative;
\item[\textup{(b)}] the multiplication operators commute, $[L_{e_2},L_{e_3}]=0$;
\item[\textup{(c)}] the polynomial $B_k$ with $k$ given by~\eqref{eq:k-abcd} defines an associative two-valued group;
\item[\textup{(d)}] $B_k$ coincides with the spectral two-valued law of $(A_{(a,b,c,d)},e_2)$, i.e.\ the spectrally determined fourth coefficient $4(a^2-bc)$ equals $4d$.
\end{enumerate}
\end{theorem}

\begin{proof}
The identity is immediate. (a)$\Leftrightarrow$(b) is Lemma~\ref{lem:MN}, (a)$\Leftrightarrow$(c) follows from the identity and Theorem~\ref{thm:universal}, and (a)$\Leftrightarrow$(d) from~\eqref{eq:spectral-k}.
\end{proof}

The equivalence (a)$\Leftrightarrow$(c) was obtained in~\cite{BKR2026}. In words: \emph{for Dubrovin's three-dimensional potential and the universal symmetric $2$-algebraic two-valued group, the WDVV equation is the associativity condition}, and the WDVV equation is exactly the statement that the spectrally determined fourth coefficient equals $4d$.

\section{The WDVV equation as a system of conservation laws}\label{sec:hydro}

\subsection{The hydrodynamic form}

Let $f$ be a solution of~\eqref{eq:WDVV3} on a domain $\Omega\subset\C^2$ and let $k_j(x,y)$ be defined by~\eqref{eq:abcd} and~\eqref{eq:k-abcd}. Equality of mixed derivatives gives
\begin{equation}\label{eq:conservation}
(k_2)_y=(k_4)_x,\quad (k_4)_y=(k_6)_x,\quad (k_6)_y=(k_8)_x,\qquad 4k_8-k_4^2+k_2k_6=0 .
\end{equation}
Eliminating $k_8=(k_4^2-k_2k_6)/4$, we obtain a system of three conservation laws. In the variables $\bu=(b,a,c)^{T}$, for which $d=a^2-bc$, it takes the quasilinear form
\begin{equation}\label{eq:hydro}
\bu_y+\cV(\bu)\,\bu_x=0,\qquad
\cV(\bu)=\begin{pmatrix}0&-1&0\\ 0&0&-1\\ c&-2a&b\end{pmatrix},
\end{equation}
that is, $b_y=a_x$, $a_y=c_x$, $c_y=(a^2-bc)_x$. This hydrodynamic form of the three-dimensional WDVV equation is an integrable Hamiltonian system of hydrodynamic type~\cite{FM1996}. Its velocity matrix has a direct meaning in terms of the spectral structure.

\begin{proposition}\label{prop:velocities}
The characteristic polynomial of $\cV(\bu)$ equals $\chi_M(\mu)=\mu^3-b\mu^2-2a\mu-c$, and $\cV(\bu)$ is conjugate to the multiplication operator $L_{e_2}$. Consequently, the characteristic velocities $dx/dy$ of~\eqref{eq:hydro} are the eigenvalues $\lambda_1,\lambda_2,\lambda_3$ of $L_{e_2}$, i.e.\ the finite branch points of the spectral two-valued law. In the real case the system is strictly hyperbolic exactly where these branch points are real and distinct.
\end{proposition}

\begin{proof}
Expanding $\det(\mu-\cV)$ along the first column gives $\mu^3-b\mu^2-2a\mu-c=\chi_M(\mu)$. The matrix $-\cV$ is a companion matrix and $e_1$ is cyclic for $M$ (Lemma~\ref{lem:MN}), so both $\cV$ and $M$ are nonderogatory; nonderogatory matrices with the same characteristic polynomial are conjugate. The last assertions follow from Proposition~\ref{prop:spectral}(iii).
\end{proof}

Thus, as the point $p$ moves over the Frobenius manifold, the parameters of the spectral two-valued group evolve according to an integrable system of conservation laws whose characteristic velocities are the branch points of the group itself.

\subsection{Reconstruction of the WDVV potential}

\begin{theorem}\label{thm:reconstruction}
Let $\Omega\subset\C^2$ be a simply connected domain, and let $k_2,k_4,k_6,k_8$ be holomorphic functions on $\Omega$ satisfying~\eqref{eq:conservation}. Then there exists a holomorphic function $f(x,y)$ on $\Omega$, unique up to addition of a polynomial of total degree at most two, such that
\begin{equation}\label{eq:k-from-f}
4f_{xxx}=k_2,\qquad 4f_{xxy}=k_4,\qquad 4f_{xyy}=k_6,\qquad 4f_{yyy}=k_8 .
\end{equation}
The function $f$ satisfies the three-dimensional WDVV equation~\eqref{eq:WDVV3}. Conversely, every normalised WDVV solution on $\Omega$ determines a solution of~\eqref{eq:conservation} by~\eqref{eq:k-from-f}.
\end{theorem}

\begin{proof}
By~\eqref{eq:conservation} the one-forms $k_2dx+k_4dy$, $k_4dx+k_6dy$, $k_6dx+k_8dy$ are closed, so on $\Omega$ they are the differentials of holomorphic functions $g_1,g_2,g_3$. Then $(g_1)_y=k_4=(g_2)_x$ and $(g_2)_y=k_6=(g_3)_x$, so $g_1dx+g_2dy=dh_1$ and $g_2dx+g_3dy=dh_2$ for some $h_1,h_2$. Since $(h_1)_y=g_2=(h_2)_x$, we have $h_1dx+h_2dy=d(4f)$ for some $f$. By construction $4f_{xx}=g_1$, $4f_{xy}=g_2$, $4f_{yy}=g_3$, which gives~\eqref{eq:k-from-f}. If two functions satisfy~\eqref{eq:k-from-f}, all third derivatives of their difference vanish, so the difference is a polynomial of degree at most two. Finally, $4k_8-k_4^2+k_2k_6=16\,(f_{yyy}-f_{xxy}^2+f_{xxx}f_{xyy})$, so the relation in~\eqref{eq:conservation} is the WDVV equation. The converse is the computation preceding~\eqref{eq:conservation}.
\end{proof}

The same statement holds, with the same proof, for real-analytic or smooth functions on a simply connected domain in $\mathbb R^2$.

\subsection{Reconstruction from the logarithm or the exponential}

For a solution $f$ of~\eqref{eq:WDVV3} and a point $p$, let $\LL_p$ and $\EE_p$ be the logarithm and the exponential (Proposition~\ref{prop:log}) of the spectral two-valued group at $p$, i.e.\ with $k=k(p)$ given by~\eqref{eq:spectral-k}. Since $k_2/4=b$, $k_4/2=2a$, $k_6/4=c$,
\begin{equation}\label{eq:Lp}
\LL_p(X)=\left(\int_0^{\sqrt X}\frac{dr}{\sqrt{1-br^2-2ar^4-cr^6}}\right)^2,
\end{equation}
and $1/\EE_p(S)=\wp(\sqrt S;g_2,g_3)+b/3$ with $g_2=8a+\frac43b^2$, $g_3=4c+\frac83ab+\frac8{27}b^3$.

\begin{proposition}\label{prop:jets}
Write $\LL_p(X)=X+\ell_1X^2+\ell_2X^3+\ell_3X^4+O(X^5)$ and $\EE_p(S)=S+\varepsilon_1S^2+\varepsilon_2S^3+\varepsilon_3S^4+O(S^5)$. Then
\[
\begin{gathered}
\ell_1=\tfrac b3,\qquad \ell_2=\tfrac{2a}5+\tfrac{8b^2}{45},\qquad \ell_3=\tfrac c7+\tfrac{52ab}{105}+\tfrac{4b^3}{35},\\
\varepsilon_1=-\tfrac b3,\qquad \varepsilon_2=-\tfrac{2a}5+\tfrac{2b^2}{45},\qquad \varepsilon_3=-\tfrac c7+\tfrac{6ab}{35}-\tfrac{b^3}{315},
\end{gathered}
\]
and conversely
\[
b=3\ell_1,\qquad a=\tfrac52\ell_2-4\ell_1^2,\qquad c=7\ell_3-26\ell_1\ell_2+20\ell_1^3,\qquad d=a^2-bc,
\]
\[
b=-3\varepsilon_1,\qquad a=\varepsilon_1^2-\tfrac52\varepsilon_2,\qquad c=-3\varepsilon_1^3+9\varepsilon_1\varepsilon_2-7\varepsilon_3 .
\]
In particular, the segment $X+\ell_1X^2+\ell_2X^3+\ell_3X^4$ of $\LL_p$ (or the corresponding segment of $\EE_p$) determines the tangent Frobenius algebra $T_p\mathcal M$ completely.
\end{proposition}

\begin{proof}
Expand the integrand in~\eqref{eq:Lp}, integrate termwise and square; the formulas for $\varepsilon_j$ follow by inverting the series, and the inverse formulas are verified by substitution.
\end{proof}

\begin{corollary}[local inverse spectral theorem]\label{cor:inverse}
Let $f$ be a solution of~\eqref{eq:WDVV3} on a simply connected domain $\Omega$. The family $p\mapsto\LL_p$ recovers the functions $a,b,c,d$ on $\Omega$ and hence, by Theorem~\ref{thm:reconstruction}, the potential $f$ uniquely modulo polynomials of degree at most two. Thus
\[
\{\LL_p\}_{p\in\Omega}\ \Longleftrightarrow\ \{\EE_p\}_{p\in\Omega}\ \Longleftrightarrow\ f(x,y)\pmod{\deg\le2}.
\]
\end{corollary}

\section{Solutions of the quantum Yang--Baxter equation}\label{sec:qybe}

\subsection{An \texorpdfstring{$R$}{R}-operator for unital algebras}

Let $\cA$ be a unital algebra over $\C$, not necessarily associative or commutative, with unit $1$ and associator $(\xi,\eta,\zeta)=(\xi\eta)\zeta-\xi(\eta\zeta)$. Let $\tau(\xi\otimes\eta)=\eta\otimes\xi$ be the flip. For constants $\alpha,\beta\in\C$ and spectral parameters $U,V\in\C$ define $R(U,V)=R_{\alpha,\beta}(U,V)\in\operatorname{End}(\cA\otimes\cA)$ by
\begin{equation}\label{eq:Rop}
R(U,V)(\xi\otimes\eta)=\alpha(U-V)\,1\otimes\xi\eta+\beta(U-V)\,\xi\eta\otimes1-(\alpha U-\beta V)\,\eta\otimes\xi .
\end{equation}
As usual, $R_{12}=R\otimes\id$, $R_{23}=\id\otimes R$ and $R_{13}=\tau_{23}R_{12}\tau_{23}$ act on $\cA^{\otimes3}$; thus $R_{ij}$ acts by~\eqref{eq:Rop} on the $i$-th and $j$-th tensor factors. The two-parameter spectral-parameter operator~\eqref{eq:Rop}, for an arbitrary unital associative algebra $\cA$, is the coloured Yang--Baxter operator of Nichita and Parashar~\cite{NP2006}; see Remark~\ref{rmk:NP} below for the relation between their construction and Theorem~\ref{thm:main}. The following theorem gives, in addition, the exact Yang--Baxter defect for an arbitrary unital algebra. Put
\[
\Delta(U,V,W)=(U-V)(U-W)(V-W),
\]
and define $\Ass_{\alpha,\beta}\in\operatorname{End}(\cA^{\otimes3})$ by
\begin{equation}\label{eq:Ass}
\Ass_{\alpha,\beta}(\xi\otimes\eta\otimes\zeta)=\alpha^2\,1\otimes1\otimes(\xi,\eta,\zeta)+\alpha\beta\,1\otimes(\xi,\eta,\zeta)\otimes1+\beta^2\,(\xi,\eta,\zeta)\otimes1\otimes1 .
\end{equation}

\begin{theorem}\label{thm:defect}
For all $U,V,W\in\C$:
\begin{enumerate}
\item[\textup{(i)}] (Yang--Baxter defect)
\begin{multline}\label{eq:defect}
R_{12}(U,V)R_{13}(U,W)R_{23}(V,W)-R_{23}(V,W)R_{13}(U,W)R_{12}(U,V)\\
=-(\alpha+\beta)\,\Delta(U,V,W)\,\Ass_{\alpha,\beta};
\end{multline}
\item[\textup{(ii)}] (regularity) $R(U,U)=(\beta-\alpha)U\,\tau$;
\item[\textup{(iii)}] (unitarity) $R(U,V)\,\tau R(V,U)\tau=(\alpha U-\beta V)(\alpha V-\beta U)\,\id$;
\item[\textup{(iv)}] (homogeneity) $R(cU,cV)=c\,R(U,V)$ for $c\in\C$.
\end{enumerate}
\end{theorem}

\begin{proof}
(ii) and (iv) are immediate from~\eqref{eq:Rop}.

(iii) We have $\tau R(V,U)\tau(\xi\otimes\eta)=\alpha(V-U)\,\eta\xi\otimes1+\beta(V-U)\,1\otimes\eta\xi-(\alpha V-\beta U)\,\eta\otimes\xi$, and by~\eqref{eq:Rop}
\begin{align*}
R(U,V)(\eta\xi\otimes1)&=(\beta-\alpha)V\,1\otimes\eta\xi+\beta(U-V)\,\eta\xi\otimes1,\\
R(U,V)(1\otimes\eta\xi)&=\alpha(U-V)\,1\otimes\eta\xi+(\beta-\alpha)U\,\eta\xi\otimes1,\\
R(U,V)(\eta\otimes\xi)&=\alpha(U-V)\,1\otimes\eta\xi+\beta(U-V)\,\eta\xi\otimes1-(\alpha U-\beta V)\,\xi\otimes\eta .
\end{align*}
Collecting terms, the coefficient of $1\otimes\eta\xi$ is $-\alpha(U-V)\bigl[(\beta-\alpha)V+\beta(U-V)+\alpha V-\beta U\bigr]=0$, the coefficient of $\eta\xi\otimes1$ is $-\beta(U-V)\bigl[\alpha(U-V)+(\beta-\alpha)U+\alpha V-\beta U\bigr]=0$, and the coefficient of $\xi\otimes\eta$ is $(\alpha V-\beta U)(\alpha U-\beta V)$.

(i) First we reduce the statement to a universal algebra. Let $\mathcal F$ be the free unital (nonassociative, noncommutative) algebra on $x,y,z$: its basis consists of $1$ and all bracketed monomials in $x,y,z$, and $1$ is a two-sided unit. For a unital algebra $\cA$ and $\xi,\eta,\zeta\in\cA$ there is a unique unital homomorphism $\varphi\colon\mathcal F\to\cA$ with $\varphi(x)=\xi$, $\varphi(y)=\eta$, $\varphi(z)=\zeta$. Since $R$ is expressed through the multiplication, the unit and the flip, $\varphi^{\otimes2}R^{\mathcal F}(U,V)=R^{\cA}(U,V)\varphi^{\otimes2}$, hence $\varphi^{\otimes3}$ intertwines $R_{ij}^{\mathcal F}$ with $R_{ij}^{\cA}$ and $\Ass^{\mathcal F}_{\alpha,\beta}$ with $\Ass^{\cA}_{\alpha,\beta}$. Therefore it suffices to verify~\eqref{eq:defect} on $x\otimes y\otimes z\in\mathcal F^{\otimes3}$, and by trilinearity it then holds on all of $\cA^{\otimes3}$. The verification is a direct expansion: each side of~\eqref{eq:defect} applied to $x\otimes y\otimes z$ is a combination of twelve basis tensors of $\mathcal F^{\otimes3}$, listed in~\ref{app:expansion}. Seven of them, which do not contain triple products, have equal coefficients on both sides; the remaining ones give
\[
\begin{aligned}
(\alpha+\beta)\Delta(U,V,W)\bigl[&\alpha^2\,1\otimes1\otimes(x(yz)-(xy)z)+\alpha\beta\,1\otimes(x(yz)-(xy)z)\otimes1\\
&+\beta^2\,(x(yz)-(xy)z)\otimes1\otimes1\bigr],
\end{aligned}
\]
which is the right-hand side of~\eqref{eq:defect} at $x\otimes y\otimes z$.
\end{proof}

\begin{corollary}\label{cor:associative}
If $\cA$ is associative, then $R_{\alpha,\beta}(U,V)$ satisfies the quantum Yang--Baxter equation with spectral parameters
\begin{equation}\label{eq:QYBE}
R_{12}(U,V)R_{13}(U,W)R_{23}(V,W)=R_{23}(V,W)R_{13}(U,W)R_{12}(U,V)
\end{equation}
for all $\alpha,\beta,U,V,W$. If $\alpha+\beta=0$, \eqref{eq:QYBE} holds for every unital algebra $\cA$.
\end{corollary}

\begin{remark}\label{rem:alpha-beta}
For $\alpha=\beta$ the family degenerates to $R(U,V)=\alpha(U-V)R_0$ with the constant operator $R_0(\xi\otimes\eta)=1\otimes\xi\eta+\xi\eta\otimes1-\eta\otimes\xi$, which by Corollary~\ref{cor:associative} satisfies the constant QYBE for associative $\cA$. For $\alpha+\beta=0$ the QYBE holds but does not detect associativity. The interesting range is therefore $\alpha\ne\pm\beta$.
\end{remark}

\subsection{The bielliptic WDVV \texorpdfstring{$R$}{R}-matrix}

Let $p$ be a point of a three-dimensional Frobenius manifold with the potential~\eqref{eq:potential}, or more generally any $(a,b,c,d)\in\C^4$, and let $A_p$ be the algebra~\eqref{eq:table} with the values~\eqref{eq:abcd} at $p$. We apply the construction of the previous subsection to $\cA=A_p$ and write
\[
\delta=U-V,\qquad A_{UV}=\beta U-\alpha V,\qquad B_{UV}=\alpha U-\beta V,\qquad \gamma=\alpha+\beta .
\]

\begin{theorem}\label{thm:main}
\begin{enumerate}
\item[\textup{(i)}] In the ordered basis $(e_1\otimes e_1,\,e_1\otimes e_2,\,e_1\otimes e_3,\,e_2\otimes e_1,\,e_2\otimes e_2,\,e_2\otimes e_3,\,e_3\otimes e_1,\,e_3\otimes e_2,\,e_3\otimes e_3)$ the operator $R^{\mathrm{bi}}_{p;\alpha,\beta}(U,V):=R_{\alpha,\beta}(U,V)$ on $A_p\otimes A_p$ has the matrix (the $j$-th column contains the image of the $j$-th basis vector)
\begin{equation}\label{eq:Rmatrix}
\resizebox{0.86\textwidth}{!}{$
\begin{pmatrix}
A_{UV}&0&0&0&\gamma a\delta&\gamma c\delta&0&\gamma c\delta&\gamma d\delta\\
0&\alpha\delta&0&(\beta-\alpha)V&\alpha b\delta&\alpha a\delta&0&\alpha a\delta&\alpha c\delta\\
0&0&\alpha\delta&0&\alpha\delta&0&(\beta-\alpha)V&0&0\\
0&(\beta-\alpha)U&0&\beta\delta&\beta b\delta&\beta a\delta&0&\beta a\delta&\beta c\delta\\
0&0&0&0&-B_{UV}&0&0&0&0\\
0&0&0&0&0&0&0&-B_{UV}&0\\
0&0&(\beta-\alpha)U&0&\beta\delta&0&\beta\delta&0&0\\
0&0&0&0&0&-B_{UV}&0&0&0\\
0&0&0&0&0&0&0&0&-B_{UV}
\end{pmatrix}$}.
\end{equation}
\item[\textup{(ii)}] Write $e_{ijk}=e_i\otimes e_j\otimes e_k$. The Yang--Baxter defect
\[
\mathcal Y=R_{12}(U,V)R_{13}(U,W)R_{23}(V,W)-R_{23}(V,W)R_{13}(U,W)R_{12}(U,V)
\]
vanishes on all basis vectors $e_{ijk}$ except
\begin{align*}
\mathcal Y(e_{223})=-\mathcal Y(e_{322})&=\gamma\Delta\,\Phi\,(\alpha^2+\alpha\beta+\beta^2)\,e_{111},\\
\mathcal Y(e_{332})=-\mathcal Y(e_{233})&=\gamma\Delta\,\Phi\,\bigl(\alpha^2\,e_{112}+\alpha\beta\,e_{121}+\beta^2\,e_{211}\bigr),
\end{align*}
where $\Delta=\Delta(U,V,W)$ and $\Phi=a^2-bc-d$.
\item[\textup{(iii)}] Consequently, if $\alpha+\beta\neq0$ and $U,V,W$ are pairwise distinct, the quantum Yang--Baxter equation~\eqref{eq:QYBE} for $R^{\mathrm{bi}}_{p;\alpha,\beta}$ holds if and only if $\Phi=0$, i.e.\ if and only if the WDVV equation holds at $p$. In particular, for every WDVV point $p$ the family $R^{\mathrm{bi}}_{p;\alpha,\beta}(U,V)$ satisfies~\eqref{eq:QYBE} for all $\alpha,\beta,U,V,W$.
\item[\textup{(iv)}] $R^{\mathrm{bi}}_{p;\alpha,\beta}(U,U)=(\beta-\alpha)U\,\tau$, $R^{\mathrm{bi}}_{p;\alpha,\beta}(U,V)\,\tau R^{\mathrm{bi}}_{p;\alpha,\beta}(V,U)\tau=-A_{UV}B_{UV}\,\id$, and $\det R^{\mathrm{bi}}_{p;\alpha,\beta}(U,V)=-A_{UV}^3B_{UV}^6$.
\end{enumerate}
\end{theorem}

\begin{proof}
(i) follows from~\eqref{eq:Rop} and~\eqref{eq:table}; for example, $R(U,V)(e_1\otimes e_1)=(\alpha\delta+\beta\delta-B_{UV})e_1\otimes e_1=A_{UV}e_1\otimes e_1$, and $R(U,V)(e_2\otimes e_2)=\alpha\delta\,e_1\otimes e_2^2+\beta\delta\,e_2^2\otimes e_1-B_{UV}\,e_2\otimes e_2$ with $e_2^2=ae_1+be_2+e_3$ gives the fifth column.
(ii) follows from Theorem~\ref{thm:defect}(i) and Lemma~\ref{lem:assoc}.
(iii) If $\gamma\Delta\ne0$ and $\Phi\ne0$, then $\mathcal Y(e_{233})\neq0$, because $e_{112}$, $e_{121}$, $e_{211}$ are linearly independent and $(\alpha,\beta)\ne(0,0)$.
(iv) The first two identities are Theorem~\ref{thm:defect}(ii),(iii). For the determinant, the subspace $W_0=1\otimes A_p+A_p\otimes1$ spanned by $e_1\otimes e_1,e_1\otimes e_2,e_1\otimes e_3,e_2\otimes e_1,e_3\otimes e_1$ is invariant. On $W_0$ the operator has the eigenvalue $A_{UV}$ on $e_1\otimes e_1$ and two blocks $\begin{pmatrix}\alpha\delta&(\beta-\alpha)V\\(\beta-\alpha)U&\beta\delta\end{pmatrix}$ of determinant $\alpha\beta\delta^2-(\beta-\alpha)^2UV=A_{UV}B_{UV}$. On the quotient $A_p^{\otimes2}/W_0$ it acts as $-B_{UV}\tau$, whose determinant is $-B_{UV}^4$. Hence $\det R=A_{UV}(A_{UV}B_{UV})^2(-B_{UV}^4)$.
\end{proof}

Thus the same relation $a^2-bc-d=0$ controls both the associativity of the tangent algebra (the WDVV equation) and the Yang--Baxter defect of the $9\times9$ matrix~\eqref{eq:Rmatrix}.

\begin{remark}[Relation to the Nichita--Parashar construction]\label{rmk:NP}
Nichita and Parashar~\cite{NP2006} constructed, for an arbitrary unital associative algebra $\cA$, the two-parameter family of coloured Yang--Baxter operators~\eqref{eq:Rop} and showed that it satisfies the quantum Yang--Baxter equation~\eqref{eq:QYBE} (Corollary~\ref{cor:associative}); that is, their construction gives the implication $\text{associativity}\Rightarrow\text{QYBE}$, for any unital associative algebra $\cA$. Theorem~\ref{thm:main} specialises this construction to the three-dimensional Dubrovin algebra $A_{(a,b,c,d)}$ and establishes a converse: for the nondegenerate family $\alpha+\beta\ne0$, the quantum Yang--Baxter equation for $R^{\mathrm{bi}}_{p;\alpha,\beta}$ is equivalent to the WDVV associativity condition $a^2-bc-d=0$. Moreover, by Section~\ref{subsec:bielliptic}, the logarithm of the associated algebraic two-valued group provides canonical Abelian spectral coordinates $U,V$ on a genus-two bielliptic curve. The essential distinction is this: the Nichita--Parashar construction gives $\cA\Rightarrow\mathrm{QYBE}$ for associative $\cA$, whereas our theorem establishes $\mathrm{WDVV}\Leftrightarrow\mathrm{QYBE}$ for the tangent Frobenius algebra and equips the resulting $9\times9$ $R$-matrix with canonical bielliptic spectral geometry.
\end{remark}

\subsection{The spectral parameter from the bielliptic Abel coordinate}\label{subsec:bielliptic}

Theorem~\ref{thm:main} holds identically in the spectral parameters. We now show that the two-valued group attached to $p$ supplies a canonical spectral parameter, which is not introduced externally but is determined by the WDVV solution itself. By~\eqref{eq:Lp}, the Abel coordinate of the spectral two-valued group at $p$ is
\begin{equation}\label{eq:abel}
U_p(X)=\sqrt{\LL_p(X)}=\int_0^{\sqrt X}\frac{dr}{\sqrt{1-br^2-2ar^4-cr^6}} .
\end{equation}

\begin{proposition}\label{prop:bielliptic}
Assume $c\neq0$ and that $P_p(\lambda)=\lambda^3-b\lambda^2-2a\lambda-c=\det(\lambda-L_{e_2}(p))$ has three distinct roots. Then the curve
\[
\Gamma_p\colon\ h^2=1-br^2-2ar^4-cr^6
\]
has genus two and is bielliptic. Its involutions $\iota_\pm(r,h)=(-r,\pm h)$ have the elliptic quotients
\[
\begin{aligned}
\Gamma_p/\iota_+&\colon\ h^2=1-bX-2aX^2-cX^3, && X=r^2,\\
\Gamma_p/\iota_-\simeq E_p&\colon\ w^2=P_p(\lambda), && \lambda=r^{-2},\ w=hr^{-3},
\end{aligned}
\]
and the second quotient is the spectral curve of the Frobenius multiplication. The holomorphic differential $dr/h$ is the pullback of $-\tfrac12\,d\lambda/w$ from $E_p$, so that $U_p$ is the Abel coordinate on $E_p$: $1/X=\wp(U_p;g_2,g_3)+b/3$.
\end{proposition}

\begin{proof}
The roots of $1-br^2-2ar^4-cr^6$ are $r=\pm\lambda_i^{-1/2}$, where $\lambda_i\ne0$ (as $c\ne0$) are the distinct roots of $P_p$; so the sextic has six distinct roots and $\Gamma_p$ has genus two. The functions $r^2$ and $h$ are invariant under $\iota_+$, and $\lambda=r^{-2}$, $w=hr^{-3}$ are invariant under $\iota_-$; since $w^2=h^2r^{-6}=\lambda^3-b\lambda^2-2a\lambda-c$, the second quotient is $E_p$. Finally $d\lambda/w=-2r^{-3}dr\cdot r^3/h=-2\,dr/h$. The last formula is Proposition~\ref{prop:log}(ii).
\end{proof}

On the $U$-line, which is the universal covering of $E_p$, the two-valued group becomes single-valued: by~\eqref{eq:law-log}, $U_p(Z_\pm)=U_p(X)\pm U_p(Y)$ for the two values $Z_\pm$ of $X*Y$. We now use the homogeneity of $R$. Put
\[
S_p(z)=R^{\mathrm{bi}}_{p;\alpha,\beta}(z,1),\qquad\text{so that}\qquad R^{\mathrm{bi}}_{p;\alpha,\beta}(U,V)=V\,S_p(U/V)\quad(V\neq0).
\]
Dividing~\eqref{eq:QYBE} by $VW^2$ we obtain the multiplicative form
\begin{equation}\label{eq:mult}
S_{12,p}(z_1)\,S_{13,p}(z_1z_2)\,S_{23,p}(z_2)=S_{23,p}(z_2)\,S_{13,p}(z_1z_2)\,S_{12,p}(z_1),
\end{equation}
with $z_1=U/V$, $z_2=V/W$, valid at every WDVV point. For any constant $\kappa$, the substitution $z_1=e^{\kappa(u-v)}$, $z_2=e^{\kappa(v-w)}$ turns~\eqref{eq:mult} into the standard difference form for $\widetilde R_p(u):=S_p(e^{\kappa u})$:
\begin{equation}\label{eq:difference}
\widetilde R_{12,p}(u-v)\,\widetilde R_{13,p}(u-w)\,\widetilde R_{23,p}(v-w)=\widetilde R_{23,p}(v-w)\,\widetilde R_{13,p}(u-w)\,\widetilde R_{12,p}(u-v).
\end{equation}
Taking for $u,v,w$ the Abel coordinates $U_p(X),U_p(Y),U_p(Z)$ of points of the two-valued group, the difference $u-v=U_p(X)-U_p(Y)$ is, up to sign, the Abel coordinate of the value $Z_-=\EE_p\bigl((U_p(X)-U_p(Y))^2\bigr)$ of $X*Y$. In this sense the spectral parameters of~\eqref{eq:difference} live on the universal Abelian covering of the spectral curve, and their composition is governed by the two-valued group. As $p$ moves over the Frobenius manifold, the coefficients $(b,a,c)$ evolve by the integrable system~\eqref{eq:hydro}, and with them the curves $\Gamma_p$ and $E_p$, the Abel coordinate $U_p$ and the family of $R$-matrices; the spectral-parameter Yang--Baxter family is transported by the WDVV solution itself.

\subsection{Commuting transfer matrices}

Theorem~\ref{thm:main} has the usual consequence for integrable lattice models; see, e.g.,~\cite{Faddeev1996}. Let $p$ be a WDVV point, $N\ge1$, and $v_1,\dots,v_N\in\C$. On $A_p^{(0)}\otimes A_p^{\otimes N}$ put $T(u)=R_{0N}(u,v_N)\cdots R_{01}(u,v_1)$ with $R=R^{\mathrm{bi}}_{p;\alpha,\beta}$, and let $t(u)=\tr_0T(u)$ be the transfer matrix acting on $A_p^{\otimes N}$.

\begin{corollary}\label{cor:transfer}
At every WDVV point, $[t(u),t(u')]=0$ for all $u,u'\in\C$. If $v_1=\dots=v_N=v$, then $t(v)$ is $\bigl((\beta-\alpha)v\bigr)^N$ times the cyclic shift of the $N$ tensor factors.
\end{corollary}

\begin{proof}
For $j=1,\dots,N$, \eqref{eq:QYBE} on the factors $0,0',j$ with parameters $(u,u',v_j)$ gives $R_{00'}(u,u')T_0(u)T_{0'}(u')=T_{0'}(u')T_0(u)R_{00'}(u,u')$. If $R_{00'}(u,u')$ is invertible, i.e.\ $A_{uu'}B_{uu'}\neq0$ by Theorem~\ref{thm:main}(iv), taking the trace over $0$ and $0'$ gives $[t(u),t(u')]=0$; the remaining values of $(u,u')$ follow by continuity. The last statement follows from the regularity $R(v,v)=(\beta-\alpha)v\,\tau$.
\end{proof}

For $\Phi\neq0$ the transfer matrices do not commute in general; for example, a direct numerical check with $N=3$ and random parameters gives $\|[t(u),t(u')]\|$ of the order of $\|t(u)\|$.

\subsection{The nearest-neighbour Hamiltonian}

The regularity of $R$ gives, by the standard argument, an integrable spin chain with nearest-neighbour interaction on $A_p^{\otimes N}$. Define $h\in\operatorname{End}(A_p\otimes A_p)$ by
\begin{equation}\label{eq:hdensity}
h(\xi\otimes\eta)=\alpha\,\xi\eta\otimes1+\beta\,1\otimes\xi\eta ,
\end{equation}
that is, $h=(\alpha\iota_1+\beta\iota_2)\circ m$, where $m$ is the multiplication, $\iota_1(\zeta)=\zeta\otimes1$ and $\iota_2(\zeta)=1\otimes\zeta$. Put
\begin{equation}\label{eq:Hchain}
H=\sum_{k=1}^{N}h_{k,k+1},\qquad h_{N,N+1}:=h_{N,1},
\end{equation}
where $h_{k,k+1}$ acts on the $k$-th and $(k+1)$-st tensor factors.

\begin{proposition}\label{prop:hamiltonian}
Let $p$ be a WDVV point, $v_1=\dots=v_N=v$ and $(\beta-\alpha)v\neq0$. Then
\begin{equation}\label{eq:logder}
t(v)^{-1}\,\frac{dt}{du}(v)=\frac{1}{(\beta-\alpha)v}\bigl(H-N\alpha\bigr),
\end{equation}
and $[H,t(u)]=0$ for all $u\in\C$. For $\Phi\neq0$, $H$ does not commute with $t(u)$ in general.
\end{proposition}

\begin{proof}
Since $R(U,V)$ is linear in $U$, $R(u,v)=(\beta-\alpha)v\,\tau+(u-v)D$ with $D=\partial_UR=(\alpha\iota_2+\beta\iota_1)\circ m-\alpha\tau$. For a regular $R$-matrix the standard computation~\cite{Faddeev1996} gives $t(v)^{-1}t'(v)=((\beta-\alpha)v)^{-1}\sum_k(\tau D)_{k,k+1}$, and $\tau D(\xi\otimes\eta)=\alpha\,\xi\eta\otimes1+\beta\,1\otimes\xi\eta-\alpha\,\xi\otimes\eta=(h-\alpha)(\xi\otimes\eta)$. Since $t(v)$ is invertible and all $t(u)$ commute (Corollary~\ref{cor:transfer}), $t(v)^{-1}t'(v)$ commutes with every $t(u)$. Formula~\eqref{eq:logder} and both commutation statements have also been checked numerically for $N=3,4$ (supplementary material).
\end{proof}

Because of the homogeneity of $R$, the Hamiltonian~\eqref{eq:Hchain} does not depend on $v$. In the basis $e_1,e_2,e_3$ it has a transparent interpretation. Regard $e_1$ as an empty site and $e_2,e_3$ as two species of particles. Then
\begin{align*}
h(e_1\otimes e_1)&=(\alpha+\beta)\,e_1\otimes e_1,\\
h(e_1\otimes e_i)&=h(e_i\otimes e_1)=\alpha\,e_i\otimes e_1+\beta\,e_1\otimes e_i\qquad(i=2,3),\\
h(e_i\otimes e_j)&=\alpha\,(e_ie_j)\otimes e_1+\beta\,e_1\otimes(e_ie_j)\qquad(i,j=2,3),
\end{align*}
with $e_2^2=ae_1+be_2+e_3$, $e_2e_3=ce_1+ae_2$, $e_3^2=de_1+ce_2$. Thus a single particle hops to the left with amplitude $\alpha$ and to the right with amplitude $\beta$, while two particles on neighbouring sites fuse into one particle, or annihilate, with amplitudes given by the structure constants of the Frobenius algebra $A_p$; the fused particle is placed on the left or on the right with weights $\alpha$ and $\beta$. In particular, $H$ never increases the number of particles, so it is block-triangular with respect to the particle number. Its particle-number-preserving part is an exclusion-type hopping operator that does not involve $a,b,c,d$; consequently the spectrum of $H$ (and likewise that of $t(u)$) does not depend on the point $p$, and the WDVV data enter through the fusion and annihilation terms, which determine the eigenvectors. The integrability of the chain, i.e.\ the existence of the commuting family $t(u)$, holds precisely at the WDVV points.

\section{An example: the \texorpdfstring{$A_3$}{A3} Frobenius manifold}\label{sec:A3}

We illustrate the constructions on the simplest nontrivial polynomial solution of the three-dimensional WDVV equation, Dubrovin's $A_3$ potential~\cite{Dubrovin1996}
\begin{equation}\label{eq:A3}
F=\tfrac12(t^1)^2t^3+\tfrac12t^1(t^2)^2+\tfrac14(t^2)^2(t^3)^2+\tfrac1{60}(t^3)^5,\qquad
f(x,y)=\tfrac14x^2y^2+\tfrac1{60}y^5 .
\end{equation}
Here
\[
b=0,\qquad a=y,\qquad c=x,\qquad d=y^2,
\]
so $\Phi=y^2-0\cdot x-y^2=0$ everywhere.

\paragraph{Spectral structure} The multiplication operator $L_{e_2}$ has the characteristic polynomial
\[
P(\lambda)=\lambda^3-2y\lambda-x ,
\]
so the spectral curve is $E_{(x,y)}\colon w^2=\lambda^3-2y\lambda-x$, and~\eqref{eq:spectral-k} gives $k=(0,4y,4x,4y^2)$. The spectral two-valued law is
\[
B(X,Y,Z)=\sigma_1^2-4\sigma_2+4y\,\sigma_1\sigma_3+4x\,\sigma_2\sigma_3+4y^2\sigma_3^2 .
\]
By~\eqref{eq:g2g3}, $g_2=8y$ and $g_3=4x$, so that $1/X=\wp(U;8y,4x)$: \emph{up to constant factors, the flat coordinates of the $A_3$ Frobenius manifold are the Weierstrass invariants of the spectral curve.} Accordingly,
\[
\disc P=32y^3-27x^2=\tfrac1{16}\bigl(g_2^3-27g_3^2\bigr),\qquad
j\bigl(E_{(x,y)}\bigr)=1728\,\frac{32y^3}{32y^3-27x^2},
\]
and the spectral curve degenerates exactly on the cusp $27x^2=32y^3$.

\paragraph{Logarithm} From~\eqref{eq:Lp},
\[
\begin{aligned}
\LL_{(x,y)}(X)&=\left(\int_0^{\sqrt X}\frac{dr}{\sqrt{1-2yr^4-xr^6}}\right)^2\\
&=X+\frac{2y}{5}X^3+\frac{x}{7}X^4+\frac{28y^2}{75}X^5+O(X^6),
\end{aligned}
\]
so that the coordinates of the point are read off from the logarithm: $y=\frac52\ell_2$, $x=7\ell_3$, in agreement with Proposition~\ref{prop:jets}. The bielliptic curve of Proposition~\ref{prop:bielliptic} is $\Gamma_{(x,y)}\colon h^2=1-2yr^4-xr^6$; it has genus two whenever $x\neq0$ and $27x^2\neq32y^3$.

\paragraph{Hydrodynamic system} The vector $\bu=(b,a,c)^T=(0,y,x)^T$ is an exact solution of~\eqref{eq:hydro}: $b_y=0=a_x$, $a_y=1=c_x$, $c_y=0=(a^2-bc)_x$. Its characteristic velocities are the three roots of $\lambda^3-2y\lambda-x=0$; in the real case the system is strictly hyperbolic in the region $32y^3>27x^2$ bounded by the cusp.

\paragraph{Quasi-homogeneity} The potential is quasi-homogeneous: $f(s^3x,s^2y)=s^{10}f(x,y)$. More generally, for any $(a,b,c,d)$ and $s\neq0$ the linear map $A_{(s^2a,\,sb,\,s^3c,\,s^4d)}\to A_{(a,b,c,d)}$ given by $e_1\mapsto e_1$, $e_2\mapsto s\,e_2$, $e_3\mapsto s^2e_3$ is an algebra isomorphism, and the corresponding $R$-matrices are conjugate: with $G=\operatorname{diag}(1,s,s^2)$,
\[
R^{\mathrm{bi}}_{(s^2a,sb,s^3c,s^4d);\alpha,\beta}(U,V)=(G\otimes G)^{-1}\,R^{\mathrm{bi}}_{(a,b,c,d);\alpha,\beta}(U,V)\,(G\otimes G).
\]
For the $A_3$ manifold, $(x,y)\mapsto(s^3x,s^2y)$ acts in exactly this way on $(a,b,c,d)=(y,0,x,y^2)$. Hence, for $y\neq0$ and up to conjugation, the $R$-matrix depends on the point only through the ratio $x^2/y^3$, that is, through the $j$-invariant of the spectral curve.

\paragraph{The $R$-matrix and the Hamiltonian} Substituting $(b,a,c,d)=(0,y,x,y^2)$ into~\eqref{eq:Rmatrix} gives
\[
R^{\mathrm{bi}}_{(x,y);\alpha,\beta}(U,V)=
\resizebox{0.74\textwidth}{!}{$
\begin{pmatrix}
A_{UV}&0&0&0&\gamma y\delta&\gamma x\delta&0&\gamma x\delta&\gamma y^2\delta\\
0&\alpha\delta&0&(\beta-\alpha)V&0&\alpha y\delta&0&\alpha y\delta&\alpha x\delta\\
0&0&\alpha\delta&0&\alpha\delta&0&(\beta-\alpha)V&0&0\\
0&(\beta-\alpha)U&0&\beta\delta&0&\beta y\delta&0&\beta y\delta&\beta x\delta\\
0&0&0&0&-B_{UV}&0&0&0&0\\
0&0&0&0&0&0&0&-B_{UV}&0\\
0&0&(\beta-\alpha)U&0&\beta\delta&0&\beta\delta&0&0\\
0&0&0&0&0&-B_{UV}&0&0&0\\
0&0&0&0&0&0&0&0&-B_{UV}
\end{pmatrix}$},
\]
which satisfies the QYBE~\eqref{eq:QYBE} at every point $(x,y)$. The Hamiltonian density~\eqref{eq:hdensity} has the fusion terms
\begin{align*}
h(e_2\otimes e_2)&=\alpha\,(ye_1+e_3)\otimes e_1+\beta\,e_1\otimes(ye_1+e_3),\\
h(e_2\otimes e_3)=h(e_3\otimes e_2)&=\alpha\,(xe_1+ye_2)\otimes e_1+\beta\,e_1\otimes(xe_1+ye_2),\\
h(e_3\otimes e_3)&=\alpha\,(y^2e_1+xe_2)\otimes e_1+\beta\,e_1\otimes(y^2e_1+xe_2).
\end{align*}
Thus two neighbouring particles of type $e_2$ fuse into a particle of type $e_3$ or annihilate with amplitude $y$, a pair $e_2,e_3$ fuses into $e_2$ with amplitude $y$ or annihilates with amplitude $x$, and two particles $e_3$ fuse into $e_2$ with amplitude $x$ or annihilate with amplitude $y^2$. The flat coordinates of the $A_3$ Frobenius manifold are therefore the coupling constants of the integrable chain~\eqref{eq:Hchain}.

\section{Concluding remarks}\label{sec:discussion}

The word ``bielliptic'' has appeared before in the Yang--Baxter literature: Martins~\cite{Martins2013} found integrable three-state vertex models whose weights lie on curves of genus five. In the construction of this paper:
\begin{enumerate}
\item the basic bielliptic curve $\Gamma_p$ has genus two, not five;
\item it is produced canonically by the logarithm~\eqref{eq:Lp} of the two-valued group attached to a WDVV solution;
\item one of its elliptic quotients is exactly the spectral curve $w^2=\det(\lambda-L_{e_2})$ of the Frobenius multiplication;
\item the same relation $a^2-bc-d=0$ controls both the WDVV associativity and the Yang--Baxter defect of the $9\times9$ matrix~\eqref{eq:Rmatrix}, for $\alpha+\beta\neq0$;
\item the family varies over the Frobenius manifold according to the integrable system of conservation laws~\eqref{eq:hydro}, whose characteristic velocities are the branch points of the two-valued group, and the logarithms of the two-valued groups reconstruct the WDVV potential;
\item the $R$-matrix defines an integrable nearest-neighbour chain whose fusion amplitudes are the structure constants of the Frobenius algebra, and which is integrable exactly at the WDVV points.
\end{enumerate}

As explained in Section~\ref{subsec:physical}, the algebra~\eqref{eq:table} itself, its structure constants, and the WDVV equation already have a well-established physical meaning as the primary chiral algebra of a perturbed two-dimensional TFT, realised concretely by the $A_3$ Landau--Ginzburg model and by the $\C P^2$ sigma-model (Remark~\ref{rmk:physical-realisations}). Items 1--6 above are additional structures built on top of this physical input; whether the resulting $R$-matrix and integrable chain admit their own physical (scattering or lattice-model) interpretation is left open.

Approaches based on one-valued formal groups compose spectral parameters by a group law; with a global rational affine spectral coordinate only the two-parameter family of the Todd genus is available~\cite{Hirzebruch1966,Bismuth1976,ColemanMcGuinness1991}. Two-valued groups remove this restriction: the universal symmetric $2$-algebraic two-valued group has four coefficients with one relation, and the relation is precisely the three-dimensional WDVV equation. It would be interesting to extend this picture to Frobenius manifolds of higher dimension, where the WDVV system consists of several equations, and to relate the corresponding $R$-matrices to multivalued groups of higher order~\cite{B2006,BK2026}.

\section*{Acknowledgements}
This work was supported by the Russian Science Foundation under grant no 25-11-00081,\\
https://rscf.ru/en/project/25-11-00081/

The author is grateful to M.I.~Kornev and V.N.~Rubtsov: the results of Sections~\ref{sec:twovalued}--\ref{sec:wdvv} are based on the joint works~\cite{BK2026,BKR2026}.

\section*{Declaration of competing interest}
The author declares that he has no known competing financial interests or personal relationships that could have appeared to influence the work reported in this paper.

\section*{CRediT authorship contribution statement}
\textbf{Victor M. Buchstaber:} Conceptualization, Methodology, Formal analysis, Writing -- original draft, Writing -- review \& editing.

\section*{Data availability}
No data were generated or analysed in this study. The computer algebra scripts used to verify the identities are provided as supplementary material.

\appendix
\section{Expansion in the proof of Theorem~\ref{thm:defect}}\label{app:expansion}

We apply both sides of~\eqref{eq:defect} to $x\otimes y\otimes z$ in $\mathcal F^{\otimes3}$. The first step on the left-hand side is
\[
\begin{aligned}
R_{23}(V,W)(x\otimes y\otimes z)={}&\alpha(V-W)\,x\otimes1\otimes yz+\beta(V-W)\,x\otimes yz\otimes1\\
&-(\alpha V-\beta W)\,x\otimes z\otimes y,
\end{aligned}
\]
and on the right-hand side
\[
\begin{aligned}
R_{12}(U,V)(x\otimes y\otimes z)={}&\alpha(U-V)\,1\otimes xy\otimes z+\beta(U-V)\,xy\otimes1\otimes z\\
&-(\alpha U-\beta V)\,y\otimes x\otimes z .
\end{aligned}
\]
Applying the remaining two factors and collecting terms, one finds that both sides contain the following seven terms with the same coefficients:
\[
\begin{array}{ll}
1\otimes z\otimes xy: & -\alpha^2(U-V)(U-W)(\alpha V-\beta W),\\
1\otimes yz\otimes x: & \alpha(U-W)(\alpha U-\beta V)(\alpha V-\beta W),\\
z\otimes1\otimes xy: & \alpha(\alpha-\beta)U(U-W)(\alpha V-\beta W),\\
z\otimes y\otimes x: & -(\alpha U-\beta V)(\alpha U-\beta W)(\alpha V-\beta W),\\
z\otimes xy\otimes1: & \beta(U-W)(\alpha U-\beta V)(\alpha V-\beta W),\\
yz\otimes1\otimes x: & \beta(\alpha-\beta)W(U-W)(\alpha U-\beta V),\\
yz\otimes x\otimes1: & -\beta^2(U-W)(V-W)(\alpha U-\beta V).
\end{array}
\]
With $\Delta=(U-V)(U-W)(V-W)$, $\gamma=\alpha+\beta$, $d_1=(U-W)(V-W)$ and $d_2=(U-V)(U-W)$, the terms containing triple products have the following coefficients on the left-hand side (LHS) and on the right-hand side (RHS):
\[
\begin{array}{lll}
 & \text{LHS} & \text{RHS}\\[2pt]
1\otimes1\otimes x(yz): & \alpha^2d_1(\beta U-\alpha V) & -\alpha^2d_1(\alpha U-\beta V)\\
1\otimes1\otimes (xy)z: & 0 & \alpha^2\gamma\Delta\\
1\otimes x(yz)\otimes1: & \alpha\beta d_1(\beta U-\alpha V) & -\alpha\beta d_1(\alpha U-\beta V)\\
1\otimes (xy)z\otimes1: & -\alpha\beta d_2(\alpha V-\beta W) & \alpha\beta d_2(\beta V-\alpha W)\\
x(yz)\otimes1\otimes1: & \beta^2\gamma\Delta & 0\\
(xy)z\otimes1\otimes1: & -\beta^2d_2(\alpha V-\beta W) & \beta^2d_2(\beta V-\alpha W)
\end{array}
\] Using $(\beta U-\alpha V)+(\alpha U-\beta V)=(\alpha+\beta)(U-V)$ and $(\alpha V-\beta W)+(\beta V-\alpha W)=(\alpha+\beta)(V-W)$, the differences LHS$-$RHS are $\pm\alpha^2\gamma\Delta$, $\pm\alpha\beta\gamma\Delta$, $\pm\beta^2\gamma\Delta$ with the plus sign for $x(yz)$ and the minus sign for $(xy)z$. This is the expression stated in the proof of Theorem~\ref{thm:defect}. For an associative algebra, where $x(yz)=(xy)z$, the two sides coincide term by term.


\end{document}